\documentclass[10pt,oneside,a4paper]{article}
\usepackage{amsmath, amsfonts, amsthm, amssymb, amscd}
\usepackage[mathcal]{euscript}  
\usepackage{verbatim}
\usepackage{mathrsfs}  
\usepackage{tensor}
\newtheorem{theorem}{Theorem}[section]
\newtheorem{proposition}[theorem]{Proposition} 
\newtheorem{lemma}[theorem]{Lemma}
\newtheorem{corollary}[theorem]{Corollary} 
\newtheorem{definition}[theorem]{Definition} 
\newtheorem{remark}[theorem]{Remark}  
\numberwithin{equation}{section} 
\newcommand \sgn {\text{sgn}}
\newcommand \xid {\dot\xi}
\newcommand \JX {J_X}
\newcommand \JY {J_Y}

\newcommand \Ecal {\mathcal E}

\newcommand \la \langle
\newcommand \ra \rangle 
\newcommand \auth 	{\textsc}   
\newcommand \Mcal 	{\mathcal M}  
 
\newcommand \Escr 	{\mathscr E}  
\newcommand \Gscr 	{\mathscr G}

\newcommand \Pscr 	{\mathscr P}  

\newcommand \Escrhat 	{\widehat {\mathscr E}}

\newcommand \RR 		{\mathbb R}   
\newcommand \eps 	\epsilon  
\newcommand \be 		{\begin{equation}}
\newcommand \ee 		{\end{equation}} 
\newcommand{\air}{N_A^\infty}
\begin{document}

\title{Weakly regular $T^2$--symmetric spacetimes. \\
The future causal geometry of Gowdy spaces}  
\author{ 
Philippe G. LeFloch\footnote{Laboratoire Jacques-Louis Lions \& Centre National de la Recherche Scientifique, 
Universit\'e Pierre et Marie Curie (Paris 6), 4 Place Jussieu, 75252 Paris, France. Email : {\sl contact@philippelefloch.org.} 
}
\, and Jacques Smulevici\footnote{D\'epartement de Math\'ematiques, Facult\'e des Sciences d'Orsay, 
Universit\'e Paris--Sud 11, 91405 Orsay, France.  
Email: {\sl jacques.smulevici@math.u-psud.fr.} 
}
}

\date{March 2014}
\maketitle

\vskip2.cm

\begin{abstract}  We investigate the future asymptotic behavior of Gowdy spacetimes on $T^3$, when the metric satisfies weak regularity conditions, so that the metric coefficients (in suitable coordinates) are only in the Sobolev space $H^1$ or have even weaker regularity.  The authors recently introduced this class of spacetimes in the broader context of $T^2$-symmetric spacetimes and established the existence of a global foliation by spacelike hypersurfaces when the time function is chosen to be the area of the surfaces of symmetry. In the present paper, we identify the global causal geometry of these spacetimes and, in particular, establish that weakly regular Gowdy spacetimes are future causally geodesically complete. This result extends a theorem by Ringstr\"om for metrics with sufficiently high regularity. We emphasize that our proof of the energy decay is based on an energy functional inspired by the Gowdy-to-Ernst transformation. In order to establish the geodesic completeness property, we prove a higher regularity property concerning the metric coefficients along timelike curves and we provide a novel analysis of the geodesic equation for Gowdy spacetimes, which does not require high-order regularity estimates. Even when sufficient regularity is assumed, our proof provides an alternative and shorter proof of the energy decay and of the geodesic completeness property for Gowdy spacetimes. 
\end{abstract}

\newpage 

\tableofcontents  

 
\section{Introduction}
\label{IN}
 
We are interested in the global causal structure of weakly regular, Gowdy spacetimes satisfying Einstein's vacuum field equations in the sense of distributions. This is the second part of a series \cite{LeFlochSmulevici0}--\cite{LeFlochSmulevici3} devoted to the definition and the analysis of $T^2$--symmetric spacetimes with weak regularity. In the first part, we defined this class of spacetimes and established an existence theory by posing the initial value problem from arbitrary initial data with weak regularity. The study of weakly regular spacetimes with symmetry was initiated in Christodoulou~\cite{Christodoulou} (for vacuum spacetimes with radial symmetry) 
and LeFloch et al.~\cite{GLF,LeFloch,LeFlochMardare,LeFlochRendall,LeFlochStewart,LeFlochStewart2}
(for vacuum or matter spacetimes with Gowdy symmetry). See also Rendall and St\aa hl \cite{RendallStahl} for a proof that singularities arise in  regular spacetimes. 

We pursue our investigation and, under the assumption of Gowdy symmetry, complete the analysis of the global geometry in the future expanding direction.  Specifically, we show that the future asymptotics and, in particular, the geodesic completeness property established earlier by Ringstr\"om \cite{Ringstrom} for sufficiently regular solutions, extends to solutions with weak regularity, as now stated.  

\begin{theorem}[Geodesic completeness of weakly regular Gowdy spacetimes]
\label{theo:1} 
The future development of a future-expanding, weakly regular, initial data set with Gowdy symmetry on $T^3$ is future timelike geodesically complete, that is, every affinely parametrized, timelike geodesic can be extended indefinitely toward the future. 
\end{theorem} 

The global future geometry of spacetimes with general $T^2$ symmetry is a much harder problem, mainly because the constraint equations do not decouple from the evolution equations. Nonetheless, there has been recently some partial progress in the study of these solutions. In particular, the asymptotics of $T^2$-symmetric polarized solutions (under a smallness assumption) are discussed in our companion paper \cite{LeFlochSmulevici3}, while related results are also derived independently by Ringstr\"om \cite{Ringstrom3}.

Apart from the appropriate notion and analysis of geodesics for weakly regular spacetimes (which we introduce in this paper), all of the terminology required in the present paper is taken from \cite{LeFlochSmulevici1}, to which the reader is refered. Observe that, as usual, the time orientation is chosen so that the spacetimes under consideration are expanding/contracting toward the future/past directions, respectively. In other words, the area of the orbits of the symmetry, $R$, which can be chosen as a time function, is increasing toward the future and decreasing toward the past. 

In the article \cite{LeFlochSmulevici1}, we generalized the Berger-Chrusciel-Isenberg-Moncrief theorem  \cite{BergerChruscielIsenbergMoncrief}  which applies to sufficiently regular spacetimes only, and we established 
 that weakly regular $T^2$--symmetric spacetimes can be covered by a global foliation in areal coordinates, that is, with a time function coinciding with the area $R \in [R_0,+\infty)$ (with $R_0 > 0$) of the surfaces of symmetry. The present paper is  concerned with the late-time behavior as $R \to +\infty$. Our analysis relies on the weak formulation of Einstein equations introduced in \cite{LeFlochSmulevici1}. Yet, as far as regular spacetimes are concerned, the present paper 
presents also an independent and self-contained method of proof. Furthermore, even for the class of weak solutions, the present paper can be read independently from the first part~\cite{LeFlochSmulevici1}, {\sl provided} the reader admits the derivation of the Einstein equations and the global existence in areal coordinates established in \cite{LeFlochSmulevici1}. (This material will be recalled at the beginning of Section~\ref{sec:2}, below.) 
Furthermore, in Section \ref{se:lie}, we also rely on a method introduced by Ringstr\"om \cite{Ringstrom} in order to analyze the asymptotic behavior of the energy in non-homogeneous spacetimes.

In order to establish Theorem~\ref{theo:1}, several difficulties must be addressed and overcome, as we now explain. 

\begin{itemize}

\item {\bf New energy functional.}
The standard method of proof developed by Choquet-Bruhat-Moncrief \cite{Choquet-Bruhat-Moncrief} and Ringstr\"om \cite{Ringstrom}, is based on the natural energy associated with the wave map system associated with the Einstein equations with Gowdy symmetry.  {\sl Correction terms,} which are controled by the energy itself via a Poincar\'e's inequality, are added in order to strengthen the dissipation of the energy functional.  We revisit this approach by presenting a more direct argument, which is expected to  be more robust in order to tackle the more challenging class of large data $T^2$--symmetric spacetimes. In the proposed method, only one metric coefficient requires a correction term, while the other essential metric component is handled by a new energy functional, inspired by the Gowdy-to-Ernst transformation\footnote{After completion of this paper, the authors learned that this technique was also recently discovered by Rinsgtr\"om.}. This is done in Section \ref{se:getdpa}.

\item {\bf Weakly regular spacetimes.}
Second, since we solely assume weak regularity conditions on the spacetime metric, this prevents us from using 
estimates on {\sl high-order derivatives,} as is done in all earlier works on Gowdy spacetimes. Instead, in our proof, we must content with  estimates at the regularity level imposed by the natural wave map energy. While the essential metric coefficients
belong to the Sobolev space $H^1$ of functions with square-integrable derivatives, other metric  coefficients have {\sl even weaker} regularity. 

\item {\bf Additional regularity along timelike curves.}
Before studying the global properties of geodesics, we must first define them and study their local existence under the low regularity assumptions on the spacetime metric used here. Indeed, our assumptions are a priori too weak to define geodesics in the usual sense, i.e. by applying the Cauchy-Lipschitz theorem 
to the usual system of ordinary differential equations governing the geodesic equation. A key observation made here is the existence of an additional regularity properties satisfied by the Christoffel symbols along timelike curves. This allows us to define geodesics and prove their existence \emph{\`a la Ascoli--Arzela} (see Section \ref{sec:4}). 

\item{\bf Global analysis of the geodesic equation.}
The standard method of proofs for geodesic completeness also uses pointwise estimates of Christoffel symbols. In this paper, we avoid the use of such estimates and rely on an ``almost monotonicity'' formula for certain components of the tangent vector to a timelike geodesic. This formula is an exact monotonicity formula provided that we are looking at geodesics which are orthogonal to the orbits of symmetry, that is geodesics with zero angular momenta. Otherwise, there are some error terms, and it remains to prove some integrability properties for them (see Section \ref{se:fgcgs}). Applied to smooth Gowdy spacetimes, this approach leads to a new and somewhat shorter proof (since no estimate are needed on second order derivatives) of the  geodesic completeness property in comparison with the proof in \cite{Ringstrom}.
\end{itemize}


\section{Decay properties for the Gowdy system}  
\label{sec:2}

\subsection{Einstein field equations in areal coordinates}

Let $(\Mcal,g)$ be a weakly regular Gowdy spacetime in the sense introduced by the authors in \cite{LeFlochSmulevici1}. From the existence theory therein (which, actually, applies even to $T^2$--symmetric spacetimes), we know that, if $R: \Mcal \to \RR$ denotes the area of the orbit of symmetry group, then its gradient vector field $\nabla R$ 
is timelike (and future oriented thanks to the standard normalization adopted in \cite{LeFlochSmulevici1}) 
and, consequently, the area can be used as a time coordinate. In 
these so-called {\sl areal coordinates,}  the variable $R$ describes the interval $[R_0, +\infty)$,
where $R_0$ is the (assumed) constant value of the area on the initial slice, 
and the metric takes the form  
\be
\label{300} 
\aligned  
g = \, & e^{2(\eta-U)} \, \big( - dR^2 + \, d\theta^2 \big) 
 +e^{2U} \, \big( dx + A \, dy \big)^2+ e^{-2U} R^2 \, dy^2.
\endaligned 
\ee
Here, the independent variables $x,y,\theta$ describe $S^1$ (the one-dimensional torus or circle)
and the three metric coefficients $U,A,\eta$ are functions of $(R,\theta)$, only. 
Note also that it will be convenient to identify $S^1$ with the interval $[0, 2\pi]$ with periodic boundary conditions. 

In \cite{LeFlochSmulevici1}, the {\sl weak version of the Einstein equations} was defined geometrically and, in areal coordinates \eqref{300},  was found to be equivalent to
the following set of evolution and constraint equations for coefficients $U, A$ in the space $C([R_0, +\infty);H^1(S^1)) \cap C^1([R_0, +\infty);L^2(S^1))$ and the coefficient $\eta \in C([R_0, +\infty);W^{1,1}(S^1)) \cap C^1([R_0, +\infty);L^1(S^1))$.
\begin{enumerate}

\item Three nonlinear wave equations for the metric coefficients $U, A, \eta$:
\begin{align}
\label{weakform1}
& (R \,  U_R)_R - (R \, \, U_\theta)_\theta 
= 2 R \, \Omega^U, 
\\ 
\label{weakform2}
& (R^{-1} \,  A_R)_R - (R^{-1}  \,  \, A_\theta)_\theta 
= e^{-2U} \Omega^A, 
\\
\label{weakform3} 
&\eta_{RR} -   \, \eta_{\theta\theta }
= \Omega^\eta , 
\end{align}
which must be understood in the sense of distributions, 
with right-hand sides  
$$
\aligned 
& \Omega^U := {e^{4U} \over 4R^2} \, \big(  \, A_R^2 - \, A_\theta^2 \big), 
\qquad
\\
&
 \Omega^A := {4 e^{2U} \over R} \, \big( -  \, U_R A_R + \, U_\theta A_\theta \big), 
\\
& \Omega^\eta := \big( -  U_R^2 +  \, U_\theta^2 \big) 
+ {e^{4U} \over 4 R^2} \big( A_R^2 - \, A_\theta^2 \big). 
\endaligned 
$$

\item Two constraint equations for the metric coefficient $\eta$: 
\begin{align}
\label{weakconstraintsr}
\eta_R &=  \, RE, 
\qquad\qquad 
\eta_\theta  = R \, F, 
\end{align}
with
$$
\aligned 
& E:= \big( \, U_R^2 + \, U_\theta^2 \big) 
  + {e^{4U} \over 4R^2} \, \big( \, A_R^2 +  \, A_\theta^2 \big),  
\qquad
\quad  F := 2 \, U_R U_\theta + {e^{4U} \over 2R^2} \, A_R A_\theta. 
\endaligned 
$$
\end{enumerate}

 
 Recall that the two equations \eqref{weakform1} and \eqref{weakform2} are the essential equations and provides a system of wave maps (with values in the hyperbolic space) which can be solved first for the metric coefficients $U$ and $A$. One recovers $\eta$ next by solving the constraint equations \eqref{weakconstraintsr}. Finally, \eqref{weakform3}  is a redundant equation, which is autoatically satisfied (in the sense of distributions). 

Important control of the {\sl first-order} derivatives of $U,A$ is deduced from the following energy functional
$$
\aligned 
& \Escr(R) := \int_{S^1} \Ecal(R,\theta) \, d\theta,
\qquad\quad
&& \Ecal :=  \Ecal^U + \Ecal^A, 
\\
& \Ecal^U :=  (U_R)^2 +  \, (U_\theta)^2, 
\quad \qquad
&&\Ecal^A := {e^{4U} \over 4R^2} \left(   (A_R)^2 +   \, (A_\theta)^2 \right). 
\endaligned 
$$
Namely, from \eqref{weakform1}--\eqref{weakform2} it follows that
\be
\label{decay1}
\aligned
&{d \over dR} \Escr(R) 
=
- {2 \over R} \int_{S^1} \left(
  \, (U_R)^2  + {e^{4U} \over 4R^2}  \,   \, (A_\theta)^2\right) \, d\theta, 
\endaligned
\ee
so that $\Escr$ is non-increasing. This observation provides us with certain {\sl time-dependent} norms of $A, U$, and $\eta$, as now stated.

\begin{proposition}[Uniform-in-time estimates for Gowdy spacetimes] 
\label{energ}
The energy inequality 
\be
\label{E400} 
\aligned 
& \Escr(R) \leq \Escr(R_0) < +\infty  
\endaligned 
\ee
holds for all $R \in [R_0, +\infty)$, together with the spacetime estimate  
\be
\label{Espacetime} 
2\, \Big\| {1 \over R}  \, (U_R)^2  \Big\|_{L^1((R_0, +\infty) \times S^1))} 
+ 
{1 \over 2}  \, \Big\| {1 \over R^3} \, e^{4U} \, (A_\theta)^2
 \Big\|_{L^1((R_0, +\infty) \times S^1))}  
\leq \Escr(R_0). 
\ee  
Furthermore, the function $\eta$ satisfies the first-order estimates (for any $R \in [R_0, +\infty)$)
\be
\label{eta404}
\aligned
& {1 \over  \, R} \| \eta_R (R, \cdot) \|_{L^1(S^1)} \leq \Escr(R) \leq \Escr(R_0), 
\\
& {1 \over R} \, \| \eta_\theta(R, \cdot)  \|_{L^1(S^1)}  \leq \Escr(R) \leq \Escr(R_0).
\endaligned
\ee
\end{proposition}

Being uniform in time, the above estimates imply certain decay properties with respect to $R$, although this approach may not lead to the optimal decay actually satisfied by solutions.  

\begin{remark}\rm  1. One easily checks from \eqref{weakconstraintsr} that (for $R \in [R_0, +\infty)$) 
\be
\label{eq:307}
{1 \over R} \, \big\| \big( e^{2 \eta(R, \cdot)} \big)_\theta \big\|_{L^1(S^1)}  
\leq 2 \, \Big\| \Ecal(R, \cdot) \, e^{2 \eta(R, \cdot)} \big\|_{L^1(S^1)},
\ee
which may be used to obtain a decay property for the spatial derivative $\eta_\theta$. 

2. For the sake of comparison, recall that most works on Gowdy spacetimes rely on a different choice of coordinates and metric coefficients, so that the metric reads 
$$
\aligned
g = \, 
& t^{-1/2} \, e^{\lambda/2} \big( -dt^2 + d\theta^2 \big) 
\\
&
             + t \, \Big( e^P \, dx^2 + 2 e^P Q dx dy + \big( e^P Q^2 + e^{-P} \big) dy^2 \Big), 
\endaligned
$$
and $P,Q,\lambda$ are the essential unknowns. The correspondence is given by 
$$
P=-2U + \ln R, \qquad
Q=A,
\qquad
\lambda= 4 \eta + 2P - \ln R.
$$ 
\end{remark}

 
\subsection{Gowdy-to-Ernst transformation and decay property for $A_R$} \label{se:getdpa}

Only a partial dissipation estimate, that is, on $U_R$ and $A_\theta$, is deduced from the standard energy, as stated in \eqref{Espacetime}. 
We establish here a first new decay estimate, specifically for $A_R$, which we are going to derive from a new energy functional. Interestingly, this argument is new even in the Gowdy case, where it can be applied to simplify the previous analysis given in \cite{Ringstrom}. We do not need a correction-type term to the main energy functional, which has the advantage to yield a sharper control and to be valid even for weakly regular spacetimes. 

To motivate our expression below, we need to recall the well-known {\sl 
Gowdy-to-Ernst transformation,} which, from any solution $U, A$ to the essential equations \eqref{weakform1}--\eqref{weakform2},
provides a new solution $U', A'$ by setting\footnote{Note, however, that $A'$ is no longer periodic in $\theta$.}
\be
\label{eq:transfo}
\aligned
U'  &:= \ln (R^{1/2}) - U,
\\
A'_R &:= e^{4U} R^{-1} A_\theta, \qquad 
\qquad
A'_\theta : =e^{4U}  R^{-1} A_R.
\endaligned
\ee
In view of the evolution equation \eqref{weakform2} satisfied by $A$, the function $A'$ is well-defined and is unique (up to an overall constant). Moreover, $U',A'$ satisfy the same evolution equations as $U,A$, that is,   
\be
\label{newU}
(R \,  U'_R)_R - (R \, \, U'_\theta)_\theta = 2 R \, \Omega^{U'}, 
\ee
\be
\label{newA} 
(R^{-1} \,  A'_R)_R - (R^{-1}  \,  \, A'_\theta)_\theta = e^{-2U'} \Omega^{A'}. 
\ee 

Motivated by \eqref{eq:transfo}, we introduce the {\sl $A$-effective energy}  (as we call it) 
\be
\label{eq:Afunc}
\Escr_b : = \int_{S^1}  \left(U_R - b R^{-1}\right)^2 +   \, (U_\theta)^2
+ {e^{4U} \over 4R^2} \left(   (A_R)^2 +   \, (A_\theta)^2 \right), 
\ee
where $b$ is an arbitrary parameter $b\in \RR$ we have combined together the energy functional associated with $U, A$ and with $U', A'$, respectively. A straightforward computation yields us 
$$
\aligned
\Escr_b 
= &\, \Escr - \frac{2b}{R} \int_{S^1} U_R + \frac{b^2}{R^2}\int_{S^1}  d\theta 
\\
= & \, \Escr - \frac{2b}{R^2} \int_{S^1} R  U_R d\theta+ \frac{b^2}{R^2}\Pscr,
\endaligned
$$
hence  
$$
\aligned
\frac{d\Escr_b}{dR} 
= & \frac{d\Escr}{dR}  + \frac{4 b }{R^2} \int_{S^1} R U_R d \theta 
- \frac{2b}{R^2} \int_{S^1}\left( (R  U_\theta )_\theta +2 R \Omega^U\right) d\theta-\frac{4 \pi b^2}{R^3}
\\
= & - {2 \over R} \int_{S^1} \left(
  \, (U_R)^2  + {1 \over 4R^2} \, e^{4U}  \,   \, (A_\theta)^2\right) \, d\theta
\\
& + \frac{4 b }{R^2} \int_{S^1} R  U_R d \theta 
- \frac{b}{R^3} \int_{S^1} e^{4U} \, \big(  \, A_R^2 -  \, A_\theta^2 \big)  d\theta-\frac{4b^2 \pi }{R^3}. 
\endaligned
$$
Therefore, we arrive at 
\be
\label{500} 
\aligned
\frac{d\Escr_b}{dR} 
= &   - {2 \over R} \int_{S^1} 
a^{-1}  \, \left(U_R - b R^{-1}\right)^2  \, d\theta
\\ &
-  \int_{S^1}  {1 \over 2R^3} \, e^{4U}  \,\left( 2b A_R^2 +   \,(1-2b) (A_\theta)^2\right) \, d\theta
 -\frac{4 \pi b^2}{R^3}
\endaligned
\ee
in which, provided $b \in [0,1/2]$, all terms have a favorable sign, so that we conclude that 
$$
\frac{d\Escr_b}{dR} \le 0. 
$$
Since $\Escr_b \geq 0$, it follows that $\Escr_b(R)$ is bounded for all $R \geq R_0$. Returning to \eqref{500} and keeping the dissipation terms of interest
we arrive at the following conclusion.

\begin{proposition}[Integral energy decay for the metric coefficient $A$] 
For all $b \in [0, 1/2]$ one has 
$$
\int_{R_0}^{+\infty} {2 \over R} \int_{S^1} {e^{4U} \over 4R^2}  \,\left( 2b A_R^2 +   \,(1-2b) (A_\theta)^2\right) \, d\theta dR 
\leq \Escr_b (R_0) + b^2 \Escr(R_0). 
$$
\end{proposition}

\

In other words, we have 
$$
\int_{R_0}^{+\infty} {\Escr^A \over R} \, dR < \infty, 
$$
with 
$$
\Escr^A :=  \int_{S^1}\,  \Ecal \, d\theta = (1 /4 R^2) \int_{S^1}\, e^{4U}  \,\left(  A_R^2 + A_\theta^2\right) \, d\theta.
$$ That is, in contrast with the earlier result in Proposition~\ref{energ}, we have now found a dissipation rate for both derivatives $A_R, A_\theta$.  It would be interesting to construct a functional providing also a control of $U_\theta$, but that the above argument does not seem to generalize in this direction.  


\subsection{Decay property for $U_\theta$} 
In order to prove energy decay for $U$, we use a small modification of the method of correction to the energy (as was used in \cite{Ringstrom}) and we check that the computation can be done at our level of regularity. This modification is performed in such a way that the relevant expressions combine well with the ones arising in our previous energy decay method (for $A$). 

More precisely, we observe as usual that the global energy dissipation bound \eqref{Espacetime} associated with the energy functional $\Escr(R)$ 
fails to provide decay for {\sl all} derivatives of $U,A$, but only for the time 
derivative $U_R$ and the spatial derivative $A_\theta$. To obtain an optimal control of 
the whole of the energy, we now introduce the {\sl $(U,A)$--effective energy functional} 
\be
\label{ModE} 
\Escrhat(R) := \Escr_b(R) +  \Gscr^U(R), 
\ee
for a fixed $b \in [0,1/2]$, 
with  
$$
\aligned
\Gscr^U &:= {1 \over R} \int_{S^1} \big( U - \la U \ra \big) \, U_R \, d\theta,   
\endaligned
$$
in which we have defined the average $\la f \ra$ of a function $f=f(\theta)$ by 
$$
\la f \ra := {1  \over 2 \pi} \, \int_{S^1} f \, d\theta.  
$$ 
The aim of the correction term is to ``trade'' a time-derivative for a space-derivative, and vice-versa.

In view of the energy identity \eqref{decay1} satisfied by $\Escr$ and focusing on the second integral term, one might expect, as $R \to +\infty$, an inequality close to 
$$
{d \over dR} \Escrhat(R) \leq - {2 \over R} \Escrhat(R) \qquad \text{ (modulo higher order terms),} 
$$
so that $\Escr$ should decay like $\frac{1}{R^2}$. 
This behavior is indeed correct for homogeneous (Gowdy or general $T^2$) spacetimes, as can be checked directly.  
However, for non-spatially homogeneous solutions, a space-derivative must be recovered from a time-derivative,  
which we can motivate as being a property of asymptotic {\sl equi-partition of energy.}
In turn, the correct rate of decay is determined by 
$$
{d \over dR} \Escrhat(R) \leq - {1 \over R} \Escrhat(R) \qquad \text{ (modulo higher order terms),} 
$$
so that $\Escrhat$ should decay like $\frac{1}{R}$. This is indeed the rate of decay established in 
Ringstr\"om~\cite{Ringstrom} for (sufficiently regular) Gowdy spacetimes.  
 
Specifically, we will use here the energy functional $\Escr_b(R)$ for $b=1/4$. The correction can be rewritten as
\be
\label{eq:corrdeco}
\Gscr^U(R)= {1 \over  R} \int_{S^1} \big( U -1/4 \ln R - \la U -1/4 \ln R \ra \big) \, 
\left(U_R - \frac{1}{4 R} \right)\, d\theta. 
\ee
We now compute the time evolution of the correction and obtain 
\be
\label{eq:Ucorre}
\aligned
\frac{d\Gscr^U(R)}{dR}
=
& -\frac{2}{R} \Gscr^U(R)
+ \int_{S^1} \left( U_R-\frac{1}{4 R}\,  \right)^2d\theta 
-\frac{1}{R}\int_{S^1} U_\theta^2
\\
& + \frac{2}{R} \int_{S^1} \Omega^U \big( U- \la U  \ra \big) d \theta
 - \frac{2 \pi}{R}\la U_R -\frac{1}{4R} \ra^2,
\endaligned
\ee
where we used\eqref{weakform1} and integration by parts in order to handle the term containing $U_{RR}$.

Thus, for $\Escrhat(R)=\Escr_{1/4}(R)+\Gscr^U(R)$, we find 
$$
\frac{d \Escrhat(R)}{dR}=-\frac{1}{R}\Escrhat(R)-\frac{1}{R} \Gscr^U(R)
+ \frac{2}{R} \int_{S^1} \Omega^U \big( U- \la U  \ra \big) d \theta
- \frac{2 \pi}{R}\la U_R -\frac{1}{4R} \ra^2. 
$$
The latter term has a favorable sign, so we have only two error terms to estimate. 
To do this, we use \eqref{eq:corrdeco}, Cauchy-Schwarz inequality, and Poincar\'e inequality, leading to 
$$
\left|\frac{1}{R} \Gscr^U(R)\right| \le 
\frac{1}{R^2} \frac{1}{4 \pi^2}\Escrhat(R). 
$$
Moreover, we can write 
\be
\label{es:omeUerr}
\aligned
\frac{2}{R} \int_{S^1} \Omega^U \big( U- \la U  \ra \big) d \theta
&  \le 
||  U- \la U  \ra ||_{L^\infty(S^1)}(R) \, \frac{\Escr^A}{R} 
\\
& \le (2 \pi)^{1/2} \, \frac{\Escrhat(R)^{3/2} }{R}, 
\endaligned
\ee
so that
\be
\label{es:etot}
\frac{d \Escrhat(R)}{dR} \le -\frac{\Escrhat(R)}{R} + K \, \frac{\Escrhat(R)}{R^2} + K \frac{\Escrhat(R)^{3/2}}{R}
\ee
for some constant $K> 0$. From the above inequality, we easily deduce that 
$$
\Escrhat(R) \le \frac{K}{R}
$$
for some constant $K>0$, {\sl provided} $\Escrhat$ is initially sufficiently small. Consequently, it remains to establish that $\Escrhat \to 0$ (without a rate) to obtain the desired $1/R$ decay for all initial data.  

Indeed, by returning to the equation \eqref{eq:Ucorre} and using \eqref{es:omeUerr} for the term $\Omega^U$, as well as the integrability of $\frac{\Escr^A}{R}$, we infer that $\frac{1}{R}\int_{S^1} U_\theta^2 d \theta$ is integrable in time. Thus, we obtain $\frac{1}{R}\Escrhat(R)$ is integrable. 
Since $\Escrhat(R)$ is decreasing, it follows that $\Escrhat \to 0$. 
We have thus reached the following conclusion. 

\begin{proposition}[Decay property of Gowdy spacetimes in areal coordinates] 
\label{energ3}
The energy inequality 
\be
\label{E4} 
\aligned 
& \Escr(R) \leq {K \over R}. 
\endaligned 
\ee
holds for all $R \in [R_0, +\infty)$ and some  constant $K>0$ depending only upon the initial data $U(R_0), A(R_0)$. 
\end{proposition} 


\section{Geodesics in weakly regular spacetimes}
\label{sec:4}

\subsection{The geodesic equation} 

Our first task is to introduce a suitably weak notion of geodesics for the spacetimes under consideration. As was observed in  \cite{LeFlochSmulevici1}, a frame $(T,X,Y,Z)$ adapted to the symmetry must be used in order to define 
the Christoffel coefficients. In addition, due to the low regularity available on the metric, these coefficients are only defined as $L^p$ functions on spacelike hypersurfaces (of the areal foliation, say). This regularity is, in principe, too limited for the geodesic equation to be well--defined.  
A key observation, made below, is that {\sl additional regularity} of the metric holds along timelike curves
and this property allows us to give a meaning to the geodesic equation.  

Recall first that, provided {\sl sufficient regularity} is assumed on the spacetime metric, a geodesic $\xi: [s_0,s_1] \to \Mcal$ with tangent vector $\xid=\xid(s)$, by definition, satisfies the parallel transport equation $\nabla_{\xid} \xid = 0$, which, in terms of its components $\xi^\alpha$ reads 
\be
\label{eq:404}
\ddot\xi^\alpha = - \big( \Gamma_{\beta\gamma}^\alpha \circ \xi \big)\, \xid^\beta \xid^\gamma. 
\ee
In \cite{LeFlochSmulevici1}, we considered a level of regularity such that the Christoffel symbols belong to some $L^p$ spaces on spacelike hypersurfaces in an adapted frame. Here, we show that in fact the additional regularity 
$\Gamma_{\beta\gamma}^\alpha \circ \xi  \in L^1(s_0, s_1)$ holds, that is to say, the Christoffel symbols admit traces in $L^1$ on timelike curves $\xi$. 

The equation \eqref{eq:404} then makes sense for the set of timelike curves $\xi$ which have at least the regularity $\xid \in L^\infty(s_0, s_1)$ and $\ddot\xi \in  L^1(s_0, s_1)$.   The existence of timelike traces for the coefficients $\Gamma_{\beta\gamma}^\alpha$ is, roughly speaking, a consequence of the fact that the metric coefficients satisfy wave equations in $1+1$ dimensions and that, by construction, timelike curves are non-characteristic for the wave equation.


\subsection{Uniformly timelike curves}

We start with the following definition. 

\begin{definition}
A curve $\xi:[s_0,s_1] \to\Mcal$ with $W^{1,\infty}$ regularity is said to be {\bf uniformly timelike} (for the physical metric) if there exists a constant $C > 0$ such that $g(\xid(s),\xid(s)) < -C$ for almost all $s \in (s_0,s_1)$.
\end{definition}

The following result follows easily from this definition. 

\begin{lemma}
\label{lem:utc}
Let $\xi \in W^{1,\infty}(s_0,s_1)$ be a uniformly timelike curve defined on a compact intervall $[s_0,s_1]$. Then, there exists a constant $D>0$ depending only on the sup-norms of $\xid$, $U \circ \xi$, and $\eta \circ \xi$ such that 
\be
\label{ineq:utc}
|\xid^R|\ge |\xid^\theta|+D, \qquad s \in (s_0,s_1). 
\ee 
\end{lemma}

\begin{proof} Indeed, from our definition and the expression of the metric, we have $|\xid^R|^2 \ge Ce^{-2\eta+2 U} + |\xid^\theta|^2$  for some $C>0$. Using the boundedness of $\eta$ and $U$ (which follows from the continuity of $\xi$ and compactness), we thus have $|\xid^R|^2 \ge C_1^2 + |\xid^\theta|^2$ for some $C_1>0$. Thus, we obtain the desired result provided that $\epsilon > 0$ is 
sufficiently 
small (depending only on the sup-norms of $\xid$, $U \circ \xi$, and $\eta \circ \xi$), so that 
$$
C_1^2 + |\xid^\theta|^2(s) \ge \left(\epsilon C_1+|\xid^\theta| (s)\right)^2$$ for all $s \in [s_0,s_1]$. 
One then easily checks that 
$$
\epsilon= -  \frac{|| \xid^\theta ||_{L^\infty(s_0,s_1)}}{C_1} + \Bigg( \frac{\| \xid^\theta \|^2_{L^\infty(s_0,s_1)}}{C_1^2}+1 \Bigg)^{1/2}
$$ works. 
\end{proof}

The following definition is motivated by the above estimate and concerns curves whose projection to the quotient spacetime $\mathcal{Q}=[R_0,\infty) \times S^1$, endowed with the (conformally equivalent) flat metric $g_\mathcal{Q}=-dR^2+d\theta^2$, are uniformly timelike. This notion will be useful in the proof of the completeness of geodesics.

\begin{definition}
A timelike curve $\xi:(s_0,s_1)$ with $W^{1,\infty}$ regularity is said to be {\bf uniformly timelike for the flat quotient metric}
 if the estimate \eqref{ineq:utc} holds globally on $(s_0,s_1)$ for some constant $D>0$.  
\end{definition}

  The proof of the following lemma is completely similar to that of Lemma \ref{lem:utc} and is omitted.

\begin{lemma}
\label{lem:qut}
Let $\xi: (s_0,s_1)$, $-\infty<s_0 <s_1 < +\infty$, be a timelike curve with $W^{1,\infty}$ regularity such that the norm of the tangent vector $g(\xid,\xid)$ is a constant $-C <0$. 
Then, the curve  $\xi$ is uniformly timelike for the flat quotient metric curve for a constant $D > 0$ arising in \eqref{ineq:utc}, which can be chosen to depend only on $\| \xid, \eta \circ \xi, U \circ \xi   \|_{L^\infty(s_0,s_1)}$ (which are finite by assumption) 
and the areal times $R(\xi(s_1)$ and $R(\xi(s_0))$.
(The latter determine the compact interval of time on which the metric functions are uniformly bounded.) 
\end{lemma}


\subsection{Regularity along timelike curves}

Our aim now is to study the regularity of the metric coefficients along timelike curves and, in particular, to establish the existence of traces. 

\begin{proposition}[Additional regularity along timelike curves]
\label{lem:408}
Let $\xi$ be a uniformly timelike $W^{1,\infty}(s_0,s_1)$ curve. Then, the following properties hold: 
\begin{enumerate}

\item The composite function $\eta \circ \xi$ belongs to $W^{1,1}(s_0, s_1)$.

\item The functions $U\circ \xi$, $A \circ \xi$ belongs to $H^1(s_0,s_1)$.

\item Denote by $\Gamma_{\alpha \beta}^{\gamma}$ the Christoffel symbols in an adapted frame. Then, there exists functions 
in $L^1(s_0,s_1)$ denoted $\Gamma_{\alpha \beta|\xi}^{\gamma}$ such that, on one hand, 
when $g$ is smooth one has $\Gamma_{\alpha \beta|\xi}^{\gamma}=\Gamma^{\alpha \beta}_{\gamma} \circ \xi$ and, on the other hand,
 for any sequence of smooth solutions ${}^\epsilon U, {}^\epsilon A$, ${}^\epsilon \eta$ to the system \eqref{weakform1}-\eqref{weakform3} satisfying 
 ${}^\epsilon U, {}^\epsilon A \to U, A$ in $C([R_0, +\infty);H^1(S^1)) \cap C^1([R_0, +\infty);L^2(S^1))$ and ${}^\epsilon \eta \to \eta$ in $C([R_0, +\infty);W^{1,1}(S^1)) \cap C^1([R_0, +\infty);L^1(S^1))$ (as $\epsilon \to 0$), one has
$$
\|^\epsilon\Gamma^{\alpha \beta}_{\gamma} \circ \xi- \Gamma_{\alpha \beta|\xi}^{\gamma} \|_{{L^1(s_0,s_1)}} \to 0 
$$
as $\epsilon \to 0$, where ${}^\epsilon\Gamma^{\alpha \beta}_{\gamma}$ denote the Christoffel symbols of $g_\epsilon$, the metric associated to ${}^\epsilon U, {}^\epsilon A$, ${}^\epsilon \eta$. In other words, the trace of $\Gamma^{\alpha \beta}_{\gamma}$ on $\xi$ exists and $\Gamma^{\alpha \beta}_{\gamma} \circ \xi \in L^1(s_0,s_1)$.

\item One has the additional regularity
$
\Gamma^{\alpha}_{ab|\xi}\in L^2(s_0,s_1) 
$
for $\alpha=R,\theta$ and $a,b=X,Y$.
\end{enumerate}
\end{proposition}

\begin{proof} Let us first we assume that our functions have enough regularity so that the traces are well defined and prove uniform estimates for them. The derivatives $\eta_R, \eta_\theta$ are essentially the energy and energy flux of the wave map system, since 
$\eta_R = a \, RE$ and $\eta_\theta  = R \, F$ with 
$$
\aligned 
& E = \big(  \, U_R^2 +  \, U_\theta^2 \big) 
  + {e^{4U} \over 4R^2} \, \big(  \, A_R^2 +  \, A_\theta^2 \big),  
\\      
& F = 2 \, U_R U_\theta + {2 e^{2U} \over R^2} \, A_R A_\theta. 
\endaligned 
$$
We need to control the term
$$
\aligned
{d \over ds} \big( \eta(\xi(s)) \big) 
& = \xid^0(s) \eta_R(\xi(s)) + \xid^1(s) \, \eta_\theta(\xi(s)) 
\\
& = R \, \Big(  \, \xid^0(R)  E(\xi(s))  + \xid^1(R) F(\xi(s)) \, \Big),  
\endaligned
$$
which is nothing but the flux of the energy equation along the timelike curve under consideration. 

Note that since $\xi$ is assumed to be timelike,  $|\xid^R| \ge\xid^\theta$. Hence, the right-hand side of the previous equation is positive. As a consequence, 
$$
\int^{s_1}_{s_0}\left |{d \over ds} \big( \eta(\xi(s)) \big) \right| ds = 
\int^{s_1}_{s_0}{d \over ds} \big( \eta(\xi(s)) \big)=\eta(\xi(s_1))-\eta(\xi(s_0)$$
is uniformly bounded in view of the regularity of $\eta$. 
In the case where our metric functions do not have enough regularity, consider as in \cite{LeFlochSmulevici1}, a sequence of solutions $^\eps U$, $^\eps A$, $^\eps \eta$ converging to our rough solution. Since $^\eps U$ and $^\eps A$ satisfies wave equations in $1+1$ dimensions, we can consider a domain $\Omega$ bounded by a initial hypersurface $R=R_0$, a characteristic hypersurface $\mathcal{H}$ and the timelike curve $\xi$. 

Apart from exchanging the time and the spatial directions, one can then repeat the estimates of Sections 6.4 and 6.5 in \cite{LeFlochSmulevici1} to prove compactness of the sequence of the traces of the solutions on the curve $\xi$. There are only minor modifications to the estimates in Sections 6.4 and 6.5 therein and so we do not repeat them here. This establishes in particular the Items 1 and 2 of the lemma.

The third claim follows since the $L^1$ norm of the Christofel symbols is controled by the $W^{1,1}$ norms of $\eta$, $U$ and $A$, as can be checked directly from the expressions of $\Gamma^{\alpha \beta}_{\gamma}$ in any adapted frame. For the last claim, it suffices to check that the expressions of the corresponding Christoffel symbols only involve derivatives of $U$ and $A$ (which are controled in $L^2$) and no derivatives of $\eta$.
\end{proof}


\subsection{Definition and existence of timelike geodesics}

In view of Lemma~\ref{lem:408}, from now on and with some abuse of notation, we denote by $\Gamma^\alpha_{\beta \gamma} \circ \xi$ 
 the traces $\Gamma^\alpha_{\beta \gamma|\xi}$.
We are in a position to establish the following result.

\begin{proposition} 
\label{prop:408}
Let $(\Mcal,g)$ be a weakly regular Gowdy spacetime. 
Given any initial point $\xi_0 \in \Mcal$, timelike vector $\xi_1 \in T_{\xi_0}\Mcal$ and initial affine parameter $s_0 \in \mathbb{R}$, there exists a value $s_1 > s_0$ and a curve  $\xi: (s_0, s_1) \to \Mcal$ such that 
 $\xid \in W^{1,1}(s_0, s_1)$,  which achieves the given data at the time $s_0$
$$
\xi(s_0) = \xi_0, \qquad \xid(s_0) = \xi_1
$$
and which satisfies the geodesic equation 
\be
\label{eq:geo}
\ddot\xi^\alpha = - \big( \Gamma_{\beta\gamma}^\alpha \circ \xi \big)\, \xid^\beta \xid^\gamma,
\ee 
where $\Gamma_{\beta\gamma}^\alpha \circ \xi \in L^1(s_0, s_1)$ denotes the trace of $\Gamma_{\beta\gamma}^\alpha$ on $\xi$ as defined in Proposition \ref{lem:408}. 
\end{proposition}

Observe that, from the proof below, it follows that the time of existence of the solution in Proposition~\ref{prop:408}
 depends only on the size of the initial vector and the norms of the metric coefficients, i.e.~the $H^1$ norm of $A, U$, the $W^{2,1}$ norm of $\eta$, etc.

We do not claim uniqueness of the solutions to the geodesic equations, which usually follows a Lipschitz estimate which is likely not hold for the Christoffel symbols under our regularity assumption. However, uniqueness is not needed to define maximal solutions to ordinary differential equations. Thus, as a corollary to the existence result of Proposition~\ref{prop:408}, there exists for every initial conditions with future pointing initial timelike vector, at least one maximal curve, defined on a interval $[s_0,s_1)$ which is a solution to the geodesic equation.  

Before giving the proof of Proposition~\ref{prop:408} and since $X,Y$ are Killing fields, we state the following result, which 
(as in the regular case) is a direct calculation, based on the fact that $X,Y$ are Killing field and the curve is a geodesic.
This observation allows us to consider just one scalar equation instead of a system of ordinary differential equations.

\begin{lemma}
For a curve $\xi$ in $W^{2,1}(s_0,s_1)$, the geodesic equation \eqref{eq:geo} is equivalent to
$$
\aligned
\frac{d}{ds}\left( g(\xid,\xid)\right)=0,  & \qquad &&  
\frac{d}{ds}\left( g(\xid,X)\right)=0,
\\
\frac{d}{ds}\left( g(\xid,Y)\right)=0, 
& \qquad &&
\ddot\xi^\theta + \big( \Gamma_{\beta\gamma}^\theta \circ \xi \big)\, \xid^\beta \xid^\gamma=0.
\endaligned
$$
\end{lemma}

From now on, it will be convenient to use the variable 
\be
\label{eqmu}
\mu := \eta-U+(1/4) \ln R. 
\ee
The following lemma will be useful in order to establish Proposition~\ref{prop:408}.
 
\begin{lemma}[$L^1$ norm of Christoffel symbols along timelike curve]\label{lem:equi}
\label{lem:409}
Let $\xi \in W^{1,\infty} (s_0, s_1)$ be a uniformly timelike curve and consider the traces of the Christoffel coefficients along this curve, that is, 
$\Gamma_{\alpha\beta}^\gamma \circ \xi$ as in Proposition \ref{lem:408}.
Then, there exists a function 
\be\label{eq:194}
\delta(s_0-s_1) \geq \| \Gamma_{\alpha\beta}^\gamma \circ \xi \|_{L^1(s_0, s_1)}  
\ee
which depends only on $D$ (introduced in \eqref{ineq:utc}), $R(\xi(s_0))$, and $R(\xi(s_1))$ uniformly with respect to $\xi$, 
and satisfies 
$$
\delta(\tau)\to 0 \text{ as } \tau \to 0.
$$ 
Moreover, for every sequence $\xi_\eps$ approaching $\xi$ in $W^{1,\infty} (s_0, s_1)$ one has 
\be
\label{eq:193}
\| (\Gamma_{\alpha\beta}^\gamma \circ \xi_\eps) \xid^\alpha_\eps - (\Gamma_{\alpha\beta}^\gamma \circ \xi) \xid^\alpha \|_{L^1(s_0, s_1)} \to 0
\quad \text{ when } \quad 
\| \xi_\eps - \xi \|_{W^{1,\infty}(s_0, s_1)} \to 0. 
\ee
\end{lemma}

\begin{proof} Recall first the expression
$
\Gamma^\theta_{\theta \theta}=\frac{1}{2}g^{\theta\theta}g_{\theta\theta, \theta}=\eta_\theta-U_\theta.
$
On the other hand, from the Einstein equations, we have
$$
\aligned
\eta_\theta-U_\theta&=
R \left( 2 \, U_R U_\theta + {e^{4U} \over 2R^2} \, A_R A_\theta\right) -U_\theta 
\\
&= R \Bigg( 2 \big( U_R- 1/(2R) \big) U_\theta + {e^{4U} \over 2R^2} \, A_R A_\theta\Bigg) 
\endaligned
$$
and, thus, 
$$
\aligned
\left|\eta_\theta-U_\theta\right|
&\le R \left( \, \big( U_R- 1/(2R) \big)^2 +  \, U_\theta^2  
  + {e^{4U} \over 4R^2} \, \big( \, A_R^2 +  \, A_\theta^2 \big)\right)
 \\
&= \frac{d}{dR}  \left( \eta-U+ (1/4) \ln R \right). 
\endaligned
$$

Observe that, by the uniform timelike property of $\xi$, there exists a constant $C>0$ (depending only on the uniform constant $D>0$ as in \eqref{ineq:utc}) such that
\begin{equation}
\label{ineq:dmubound}
|\eta_\theta(\xi(s))-U_\theta(\xi(s))| 
\le C \frac{d}{ds}\left( \left( \eta-U+ (1/4) \ln R \right)(\xi(s))\right)
= C \frac{d}{ds}\left( \mu(\xi(s))\right).
\end{equation}
Hence, it follows that the $L^1(s_0,s_1)$ norm of $\Gamma^\theta_{\theta \theta}(\xi(s))$ can be uniformly bounded by a constant depending only on $D>0$ and the $L^\infty$ norms of $U$ and $\eta$ in $[R(\xi(s_0)),R(\xi(s_1))] \times S^1$. The proof of \eqref{eq:194} for the other Christoffel symbols is similar.

Furthermore, the estimates \eqref{eq:193} follows from a standard energy eargument in a domain bounded by two timelike curves $\xi_1$, $\xi_2$ and two spacelike hypersurfaces.
\end{proof}

\begin{proof}[Proof of Proposition~\ref{prop:408}] 
Let $\xi_0$, $\xi_1$ be initial conditions as in the statement of the proposition, and consider the constants
$$
J_X := g(\xi_1,X),
\qquad
J_Y := g(\xi_1,Y),
\qquad 
-N^2 := g(\xi_1,\xi_1).
$$
Denote by $^{\mathcal{T}}g$ the metric induced by $g$ on each 2-torus associated to a fixed $(R,\theta)$.
Given any point $p \in \Mcal$, we denote by $(\dot \Xi^X(p),\dot \Xi^Y(p))$ the unique (vectorial) solution to the algebraic system
$$
\aligned
& ^{\mathcal{T}}g\left( \left( \dot \Xi^X(p), \dot\Xi^Y(p) \right) , X \right) = \JX,
\qquad 
&&^{\mathcal{T}}g\left(\left(\dot \Xi^X(p),\dot \Xi^Y(p) \right),Y\right) = \JY,
\\
& \sgn(\dot\Xi^X(p))= \sgn(\xi_1^X),
\qquad && \sgn(\dot\Xi^Y(p))= \sgn(\xi_1^Y). 
\endaligned
$$
Given $\xid^\theta \in \mathbb{R}$ and a point $p \in \Mcal$, let $\dot \Xi^R( \xid^\theta, p)$ be such that with $\xid=\left( \dot \Xi^R( \xid^\theta, p),  \xid^\theta, \xid^X(p), \xid^Y(p) \right)$,
$$
g(\xid, \xid) = -N^2, 
\qquad 
\sgn(\dot\Xi^R) = \sgn(\xi_1^R). 
$$

Finally, given $\xid \in L^1(s_0,s_1)$, we associate the curve
$
\xi :=\xi_0+ \int_{s_0}^s \xid(s) \, dt.
$
Given $(s_0,s_1)$, we consider now the following mapping defined on the set of $W^{2,1}$ curves: 
$\psi: W^{2,1}(s_0,s_1) \ni \psi \to \psi(\xi) \in W^{2,1}(s_0,s_1)$
such that 
$$ 
\dot \psi(\xi)^X(s)=\dot\Xi^X(\xi(s)), 
\qquad 
\dot \psi(\xi)^Y(s)=\dot\Xi^Y(\xi(s)), \qquad \dot \psi(\xi)^R(s)=\dot\Xi^R(\xi(s)),
$$
and 
$$
\dot \psi(\xi)^\theta(s) := \xi_1^\theta - \int_{s_0}^s \big( \Gamma_{\alpha\beta}^\theta \circ \xi \big)\, \xid^\alpha \xid^\beta \, dt, 
\quad 
\psi(\xi^\alpha)(s) := \xi_0^\alpha + \int_{s_0}^s \dot \psi(\xi^\theta)(s) \, dt. 
$$
Denote by $B(s_0,s_1)$ the ball of radius $1$ about the curve $\xi_0+\xi_1 (s-s_0)$ with respect to the $W^{2,1}(s_0,s_1)$ norm. Let $\widetilde{B}(s_0,s_1)$ be the intersection of $B(s_0,s_1)$ with the set of $W^{2,1}(s_0,s_1)$ curves $\xi$ such that $g(\xid, \xid)=-N^2$. Proposition~\ref{prop:408} then follows provided we can establish the following result. 

\vskip.15cm 

\noindent{\bf Claim.} 
There exists some $\epsilon > 0$ such that if $|s_1-s_0| \le \epsilon$, $\psi(\widetilde{B}(s_0,s_1)) \subset \widetilde{B}(s_0,s_1)$. Moreover, there exists a sequence $(\xi_k)_{k\in \mathbb{N}}$ of curves in $\widetilde{B}(s_0,s_1)$ which converges in $C^1(s_0,s_1)$ as $k \to +\infty$ to a curve $\xi \in W^{2,1}(s_0,s_1)$ which is a fixed point of $\psi$, i.e. a solution to the geodesic equation \eqref{eq:geo}. 

\vskip.15cm 

Note first that from the uniform bounds on the tangent vectors of the curves lying in $\widetilde{B}(s_0,s_1)$, the quantity $\sup\big\{ R(\xi(s)) : \xi \in \widetilde{B}(s_0,s_1) \big\}$ is finite. Hence, all the curves stay in a compact region of $\Mcal$. As a consequence, since $U$, $A$, $\eta$ are continuous, they are uniformly bounded along any of the curves lying in $\widetilde{B}(s_0,s_1)$ (with the bounds independent of the curve considered.)
From the condition $g(\xid, \xid)=-N^2$ and the sup norm bound on $|\xid|$, it then follows that all curves in $\widetilde{B}(s_0,s_1)$ are uniformly timelike and that the same constant $D$ in Lemmas \ref{lem:qut} and \ref{lem:equi} can be taken for all curves.
For $\xi \in \widetilde{B}(s_0,s_1)$ and using $|| \xid - \xi_1 || < 1$, we have 
\be
 \label{eq:cres}
\int_{s_0}^{s_1}\left|\frac{d^2}{ds^2}\psi(\xi)^\theta(s)\right|ds \le C || \Gamma_{\alpha \beta}^{\theta} \xid^{\alpha} \xi^{\beta} ||_{L^1(s_0,s_1)} \le C \delta,
\ee
for any given $\delta > 0$, provided $|s_1-s_0|$ is chosen to be sufficiently small.  It follows that for $\delta > 0$ sufficiently small, we can ensure that $\psi(\widetilde{B}(s_0,s_1)) \subset \widetilde{B}(s_0,s_1)$. Consider now a sequence of curves $\xi_k=\psi^k(\xi)$ where $\xi$ is any curve in $\widetilde{B}(s_0,s_1)$. In view of Lemma \ref{lem:equi}, an estimate similar to \eqref{eq:cres} holds when $s_0$ et $s_1$ are replaced by any $s_0'<s_1'$ with $s_0',s_1' \in [s_0,s_1]$. As a consequence, there exists a function $\delta(s)$ with $\delta(s) \rightarrow 0$ as $s \rightarrow 0$ such that

$$
\int_{s_0'}^{s_1'} |\ddot \xi_k(s)| ds \le \delta_k(s_1'-s_0'). 
$$
Moreover, the function $\delta_k(s)$ can be expressed as algebraic function of $R \circ \xi_k$, $U \circ \xi_k$, and $\eta \circ \xi_k$ (recalling \ref{ineq:dmubound}) which are all $C^1$ functions, with $C^1$ norm uniformly bounded with respect to $k$. Thus, one can replace $\delta_k(s)$ by a uniformly continous function on $[s_0,s_1]$ which is independent of $k$.
Equicontinuity of the sequence of $\dot \xi_k$ follows. Since we have already established uniform boundedness, an application of the Arzela-Ascoli theorem gives us the existence of a converging subsequence (in $C^1$). That the limit is solution to the geodesic equation is then a consequence of \eqref{eq:193} of Lemma \ref{lem:equi}.
\end{proof} 


\section{Future geodesic completeness of Gowdy spacetimes}  \label{se:fgcgs}

\subsection{$L^\infty$ estimates for $U$, $Ae^{2U}$ and $\rho$} \label{se:lie}

As in in the regular case, given a maximal solution to the geodesic equation, we say that the geodesic is future complete if the solution is global to the future, i.e. the interval of definition of the curve is of the type $[s_0,+\infty)$.
Let now $\xi$ be a future directed maximal solution defined on an interval $[s_0,s_1)$ to the geodesic equation as constructed in the previous section. We will denote by $J_X$ and $J_Y$ the conserved angular momenta $g(\xid,X)$ and $g(\xid,Y)$ respectively and by $K >0$ the magnitude of $\xid$, i.~e.~$-K^2=g(\xid,\xid)$.
 
The following estimates are standard in view of the avalable energy decay estimate. 

\begin{lemma} \label{lem:reu}
There exists a constant $C$ depending only on the norm of initial data of the solution such that, for all $R \ge R_0$ and uniformly in $\theta \in S^1$,
$$
\aligned
|U|(R,\theta) &\le C R^{1/2},
  \\
|A e^{2U}|(R,\theta)  &\le C e^{C R^{1/2}}.
\endaligned
$$
\end{lemma}

\begin{proof}
For $\theta,\theta' \in S^1$, we have
$$
U(\theta,R)=U(\theta',R)+\int_\theta'^\theta U_\theta(\theta'',R) d \theta''
$$
and, by integration in $\theta'$, we obtain
$$
U(\theta,R)=\frac{1}{2\pi} \int_{S^1}U(\theta',R)d\theta'+ \frac{1}{2\pi}\int_{S^1}\left(\int_{\theta'}^\theta U_\theta(\theta'',R) d \theta''\right)d\theta'. 
$$
Using the Cauchy-Schwarz inequality, we obtain
$$
|U(\theta,R)| \le \left|\frac{1}{2\pi} \int_{S^1}U(\theta',R)d\theta'\right|+(2\pi)^{1/2}\left(\int_{S^1} | U_\theta |^2(\theta'',R) d \theta''\right)^{1/2}.
$$
The second term in the right-hand side behaves like $R^{-1/2}$ thanks to the decay of the energy, while the first term in the right-hand side can be estimated using the energy decay as follows. 
 First, we write 
$$
\left|\int_{S^1}U(\theta',R)d\theta'\right|\le\left|\int_{S^1}\int_{R_0}^{R}U_R(\theta',R')dR'd\theta'\right|+\left|\int_{S^1}U(\theta',R_0)d\theta'\right| 
$$
with the second term in the right-hand side being controlled by the $H^1$ norm of the data for $U$. 
For the first term in the right-hand side, we use again the Cauchy-Schwarz inequality: 
$$
\aligned
 \left|\int_{S^1}\int_{R_0}^{R}U_R(\theta',R')dR'd\theta'\right| 
& \le (2\pi)^{1/2}\int_{R_0}^{R}\left|\int_{S^1} U_R^2(\theta',R')d\theta'\right|^{1/2}dR' 
\\
&\le  K \int_{R_0}^{R}(R')^{-1/2} dR' \le  K R^{1/2}.
\endaligned
$$

Next, we derive the estimates on $A e^{2U}$. As before, for all $\theta,\theta' \in S^1$ and for all $R \ge R_0$, we first write
$$
A e^{2U} (R,\theta)=A e^{2U} (R,\theta')+\int_{\theta'}^{\theta} \left(A e^{2U}\right)_\theta (R,\theta'')d\theta''
$$
and we then integrate in $\theta'$ on $S^1$ in order to obtain
\begin{equation}\label{eq:aeut}
2\pi A e^{2U} (R,\theta)=\int_{S^1}A e^{2U} (R,\theta')d\theta'+ \int_{S^1}\int_{\theta'}^{\theta} \left( A e^{2U}\right)_\theta (R,\theta'')d\theta''d\theta'.
\end{equation}
To estimate the first term in the right-hand side, we use
$$
\aligned
\int_{S^1} Ae^{2U} (R,\theta') d\theta'
=&\int Ae^{2U} (R_0,\theta') d\theta'+\int_{R_0}^R \int_{S^1} A_R e^{2U} (\theta',R') d\theta'dR'
\\
&+\int_{R_0}^R \int_{S^1} A 2 U_R e^{2U} (\theta',R') d\theta'dR',
\endaligned
$$
from which it follows that
$$
\aligned
&\left| \int_{S^1} Ae^{2U} (R,\theta') d\theta'\right|
\\
& \le 
 D + \int_{R_0}^R (2\pi)^{1/2}\left( \int_{S^1} A_R^2 e^{4U} (\theta',R') d\theta'\right)^{1/2}d\theta'dR' 
\\
&\quad +\int_{R_0}^R \left( \int_{S^1} A^2 e^{4U} (\theta',R') d\theta' \right)^{1/2} \left( \int_{S^1} U_R^2(\theta',R') d\theta'\right)^{1/2}dR'. 
\endaligned
$$
Here, we have used the Cauchy-Schwarz inequality twice and $D>0$ is a constant depending only on the norm of the initial data.

By setting $N_A^\infty(R):=||Ae^{2U}(R,\theta)||_{L^\infty(S^1)}$, it follows from the previous inequality and the energy decay that there exists some constant $D >0$ depending only on the norm of the initial data such that, for all $R \ge R_0$: 
\begin{eqnarray*}
\int_{S^1} Ae^{2U} (R,\theta') d\theta' &\le& \left( R^{3/2}+1 \right) D + \int_{R_0}^R \air(R') R'^{-1/2} dR'. 
\end{eqnarray*}
Next, we estimate the second term on the right-hand side of \eqref{eq:aeut} as follows: 
\begin{eqnarray*}
\left| \int_{\theta'}^{\theta} \left(A e^{2U}\right)_\theta (R,\theta'')d\theta''\right| \le \int_{S^1} \left|A_\theta\right| e^{2U}(R,\theta'')d\theta'' + \int_{S^1} \left|A e^{2U} 2 U_\theta\right| (R,\theta'')d\theta'',
\end{eqnarray*}
from which, by applying the Cauchy-Schwarz inequality and using the energy decay, we find 
$$
\left| \int_{\theta'}^{\theta} \left(A e^{2U}\right)_\theta (R,\theta'')d\theta''\right| \le C R^{1/2}+ C R^{-1/2} \air(R) 
$$
for some constant $C>0$ depending only on the norm of the initial data. By combining with our previous estimates, we obtain 
$$
\air(R) \le C R^{1/2} + C R^{-1/2} \air(R)+ D\left( 1 +R^{3/2} \right) + D \int_{R_0}^{R} \air(R') R'^{-1/2} dR'.
$$
To conclude, we note that if $R$ is sufficiently large, we can absorb the second-term on the right-hand side to the left and then apply Gronwall inequality. (Note that on any bounded intervall of time, $\air$ can be estimated directly using the finiteness of the energy.)
\end{proof}

On top of these pointwise estimates on $U$ and $Ae^{2U}$, we will also need the following knowledge on the quantity $\rho:=\eta-U$. For this, let us recall Theorem 1.7 of \cite{Ringstrom}, rewritten here in our notation. 

\begin{theorem}[Ringstr\"om~\cite{Ringstrom}]\label{th:ring17}
Assume that the metric is smooth and is non-homogeneous. Let $\rho:=\eta-U$. Then, there exists some constant $c>0$ such that 
$$
\big\| \frac{d \rho}{dR} - c \big\|_{L^1(S^1)}(R) \lesssim {1 \over R}.
$$
As a consequence, there exists constant $C > 0$, $c>0$ such that for all $R \ge R_0$ and for all $\theta \in S^1$, the following lower bound holds
$$
e^{\rho} (R,\theta) \ge C e^{c R}.
$$
\end{theorem}

Now we claim that the statement remains true under our weak regularity assumptions. Indeed, Theorem 1.7 is itself a direct consequence of Theorem 1.6 in \cite{Ringstrom} and one can check that the proof of Theorem 1.6 can be essentially repeated for our class of metrics. Indeed, as can been directly checked from the proof of Theorem 1.6 in Section 9 of \cite{Ringstrom}, all second-order derivatives of the metric appearing in the computation always occur as either a total derivative (and thus can be transformed to boundary terms controlled by the initial norm of our solutions after integration) or in combinations so that one can replace them using the Einstein equations. Thus, all the estimates in this section can be written using only the energy norms of $U$ and $A$ introduced in Section 2.
Thus, we have the following conclusion. 

\begin{corollary}\label{lem:ring17w}
The conclusion in Theorem \ref{th:ring17} holds for weakly regular Gowdy spacetimes.
\end{corollary}
 

\subsection{Estimates based on the angular momentum}

Recall that we have conservation along $\xi$ of the following two angular momenta
$$
\aligned
& \JX(s) := g(\xid,X)=\JX(s_0), 
\\
& \JY(s) := g(\xid,Y)=\JY(s_0).
\endaligned
$$
Hence, we have the following identities for $\xid^X$ and $\xid^Y$.

\begin{lemma}The following identities hold
$$
\aligned
e^{2 U}\left( \xid^X +A \xid^Y \right)(s)&=\JX(s_0), 
\\
e^{2 U} A \left( \xid^X +A \xid^Y \right)(s) + e^{-2 U}R^2 \xid^Y(s)&=\JY(s_0),
\endaligned
$$
which can be inverted and provide 
$$
\aligned
\xid^X&=\left( R^{-2} A^2 e^{2 U} + e^{-2 U} \right) J_X - A e^{2U} R^{-2} J_Y,\\
\xid^Y&=-A e^{2 U } R^{-2}J_X + R^{-2} e^{2 U} J_Y.
\endaligned
$$
\end{lemma}

As a direct consequence of the estimates of the previous section, we then have
\begin{corollary}
There exists a constand $C$ depending only on the norm of initial data for the metric and the values of $J_X$ and $J_Y$ such that, for all $s \in [s_0,s_1)$,
$$
|\xid^X(s), \xid^Y(s)| \le C \exp{\left(CR^{1/2} \right)}.
$$
\end{corollary}

As in \cite{Ringstrom}, we introduce the vectors
$e_2=e^{-U} X$ and $e_3=e^{U} R^{-1} \left( Y-AX \right)$, and one can easily check that $g(e_2,e_2)=1=g(e_3,e_3)$ and $g(e_2,e_3)=0$.
Then, define the components $f_2:=g(e_2,\xid)$ and $f_3:=g(e_3,\xid)$. 
In \cite{Ringstrom}, the estimate $1+f_2^2+f_3^2 \le K \exp(Kt^{1/2})$
is established from bounds on second-order derivatives of the metric. Here, we have provided an alternative proof using only first-order derivative of the metric.


\subsection{The structure of the evolution equation for $\xid^R$}

The aim of this section is to clarify the structure of the evolution equation for $\xid^R$.

\begin{lemma}The following differential inequality holds
\be
\label{den50} 
\frac{d}{ds} \left( R^{-1/2} \xid^R  \right)\le - R^{-1/2}\Gamma^R_{ab} \xid^a\xi^b. 
\ee
\end{lemma}

\begin{proof}
From the geodesic equation, we have
$$
\aligned 
\ddot \xi^R 
= & -\Gamma^R_{\alpha \beta} \xid^\alpha \xid^\beta 
\\
=& -\Gamma^R_{RR} \left(\xid^R\right)^2 -2\Gamma^R_{R\theta} \xid^R\xi^\theta -\Gamma^R_{\theta \theta} \left(\xid^\theta\right)^2  -\Gamma^R_{ab} \xid^a\xi^b, 
\endaligned
$$
where $a,b$ range over all possible combinations of $X,Y$ and we have used the fact that $\Gamma^R_{Ra}=\Gamma^R_{\theta a}=0$, for $a=X,Y$. We recall the following formula for the Christoffel symbols: 
$
\Gamma^{\alpha}_{\beta \gamma}=\frac{1}{2} g^{\alpha \rho} \left( g_{\rho \beta,\gamma}+g_{\gamma \rho, \beta}+g_{\beta \gamma,\rho} \right)
$: 
\begin{align*}
\Gamma^{R}_{RR}&=\eta_R -U_R, \qquad \Gamma^{R}_{R\theta}=\eta_\theta-U_\theta,\\
\Gamma^{R}_{\theta \theta}&=\eta_R -U_R, \quad \Gamma^{R}_{ab}\,=-\frac{1}{2}e^{-2(\eta-U)}g_{ab,R}.
\end{align*}
In view of the following computation for the Christoffel symbols, we obtain 
$$
\aligned
\ddot \xi^R 
=& -\left( \eta_R -U_R \right) (\xid^R)^2 -2 \left( \eta_\theta -U_\theta \right) \xid^\theta \xid^R
\\
&-\left( \eta_R -U_R \right) (\xid^\theta)^2-\Gamma^R_{ab} \xid^a\xi^b
\endaligned
$$

Recall now the equations for $\eta_R$ and rewrite $\eta_R-U_R$ as follows
$$
\aligned
\eta_R-U_R
=& R \left( U_R^2+ U_\theta^2 \right)+ \frac{e^{4U}}{4R^2} \left( A_R^2+A_\theta^2 \right)-U_R 
\\
=&  R \left( \left(U_R-\frac{1}{2R} \right)^2+ U_\theta^2 \right)+ \frac{e^{4U}}{4R^2} \left( A_R^2+A_\theta^2 \right)-\frac{1}{4R}.
\endaligned
$$
Similarly, we have
$$
\eta_\theta -U_\theta= R \left( 2 \left( U_R-\frac{1}{2R}\right) U_\theta +{e^{4U} \over 2R^2} \, A_R A_\theta \right),
$$
and 
it follows from the Cauchy-Schwarz inequality that
$$
|\eta_\theta -U_\theta| \le \eta_R-U_R+\frac{1}{4R}.
$$

Returning to the geodesic equation for $\xid^R$, we have
$$
\aligned
\ddot \xi^R 
=& -\left( \eta_R -U_R +\frac{1}{4R}\right) (\xid^R)^2 -2 \left( \eta_\theta -U_\theta \right) \xid^\theta \xid^R
\\
&-\left( \eta_R -U_R +\frac{1}{4R}\right) (\xid^\theta)^2+\frac{1}{4R} \left(\xid^R)^2+(\xid^\theta)^2 \right)
-\Gamma^R_{ab} \xid^a\xi^b 
\\
\le& \frac{1}{4R} \left((\xid^R)^2+(\xid^\theta)^2 \right)-\Gamma^R_{ab} \xid^a\xi^b 
\\
\le&\frac{1}{2R} (\xid^R)^2-\Gamma^R_{ab} \xid^a\xi^b, 
\endaligned
$$
where we have observed that the matrix 
$
\left( \begin{array}{cc}\eta_R -U_R +\frac{1}{4R} &  \eta_\theta -U_\theta \\
\eta_\theta -U_\theta &\eta_R -U_R +\frac{1}{4R} 
\end{array}\right)
$
is negative definite in order to derive the first inequality above, and in the last inequality we have used $(\xid^\theta)^2 \le (\xid^R)^2$, which is a consequence of the timelike property of the curve. 

Now, since $\frac{d R}{ds}=\xid^R$, we obtain
$
\ddot \xi^R\le \frac{1}{2}\frac{d }{ds}\left( \ln R \right) \xid^R-\Gamma^R_{ab} \xid^a\xi^b,
$
which is equivalent to \eqref{den50}. 
\end{proof}


\subsection{Christoffel symbols and first derivatives of the metric}

Next, we estimate the right-hand side of \eqref{den50}. From the expression of the Christoffel symbols, the following estimate is immediate. 
in which we use our notation $\mu$ as well as the variable
\be
\rho:=\eta-U.
\ee

\begin{lemma} One has 
$$
\Gamma^R_{ab} \le C e^{-2 \rho} \left( \left(\mu_R\right)^{1/2}+1 \right).
$$
\end{lemma}

Now we give an estimate of $\mu_R$ in terms of a total derivative along the curve.

\begin{lemma} The  estimate
$$
0 \leq \mu_R
\le \frac{ 2 \xid^R e^{2 \rho}}{K^2+ J_X^2 e^{-2U}+(J_Y-A J_X)^2 e^{2 U }R^{-2}}  \, \frac{d\mu}{ds}
$$
holds and, consequently the (weaker) estimate
\begin{equation} \label{ineq:murw}
\mu_R 
\le K^{-2} e^{2 \rho} 2 \xid^R \frac{d \mu}{ds}  
\end{equation}
is satisfied. 
\end{lemma}

\begin{proof}
From the conservation of the length of the tangent vector to the timelike geodesic, we have
$$
\aligned
-K^2 e^{-2(\eta-U)} +\left(\xid^R\right)^2
& -e^{-2(\eta-U)} e^{2U} \left( \xid^X + A \xid^Y \right)^2
\\
& - e^{-2(\eta-U)} e^{-2U} R^2 \left( \xid^Y \right)^2 = \left(\xid^\theta\right)^2
\endaligned
$$ 
and 
$$
\frac{d \rho}{ds}
=\left(\eta-U \right)_R \xid^R +\left(\eta-U \right)_\theta \xid^\theta
=\mu_R \xid^R +\left(\eta-U \right)_\theta \xid^\theta- \frac{1}{4R}\xid^R .
$$
Note, moreover, that
$$
\aligned
\left(\frac{ \xid^\theta}{ \xid^R} \right)^2 
&= 1-\frac{ e^{-2 \rho}}{(\xid^R)^2}\left(K^2+ J_X^2 e^{-2U}+(J_Y-A J_X)^2 e^{2 U }R^{-2} \right) 
\\
&\le 
1- \frac{K^2 e^{-2 \rho}}{(\xid^R)^2}.
\endaligned
$$

Setting $\chi=\frac{K^2 e^{-2 \rho}}{(\xid^R)^2}$, we then have
$$
\aligned
\frac{d \rho}{ds}
& =\left(1 - (1-\chi)^{1/2}\right)\mu_R \xid^R - \frac{1}{4R}\xid^R+ (1-\chi)^{1/2} \, \mu_R \, \xid^R 
 + \rho_\theta \xid^\theta 
\\
& \geq \left(1-(1-\chi)^{1/2}\right) \, \mu_R  \xid^R - \frac{1}{4R}\xid^R, 
\endaligned
$$
by 
using that $\rho_R+\frac{1}{4R} \ge |\rho_\theta|$ and $\rho_R+\frac{1}{4R} \ge 0$.
Hence, we obtain 
$$
\mu_R
\leq
 \frac{1}{ {\dot{\xi^R}} }(1-(1-\chi)^{1/2})^{-1}  \frac{d\mu}{ds}.
$$
Recall now that for all $0 \le \alpha \le 1$, we have $(1-\alpha)^{1/2} \le 1-\frac{\alpha}{2}$.
Hence, we may estimate 
$$
1- \left( 1-\frac{K^2 e^{-2 \rho}}{(\xid^R)^2}\right)^{1/2} \ge \frac{K^2 e^{-2 \rho}}{2 \left( \xid^R\right)^2}
$$
and we obtain the weaker estimate in the lemma:  
$$
\rho_R+\frac{1}{4R} \le K^{-2} e^{2 \rho} 2 \xid^R \frac{d \mu}{ds}.  
$$
The stronger estimate is derived by the same method, but using now  
$$
\chi=\frac{ e^{-2 \rho}}{(\xid^R)^2}\left(K^2+ J_X^2 e^{-2U}+(J_Y-A J_X)^2 e^{2 U }R^{-2}\right).
$$ 
\end{proof}

In view of the $L^\infty$ estimates on $U$, we now reach the following conclusion. 

\begin{corollary} There exists a constant $C \ge 0$, depending only on the norm of the initial data, such that, along $\xi(s)$, the following inequality holds
  \begin{equation} \label{ineq:fin}
\mu_R \le \frac{e^{2 \rho}}{K^2+C \exp\left({-C R^{1/2}}\right)} 2 \xid^R \frac{d \mu}{ds}.
\end{equation}
Moreover, the constant $C$ appearing in the denominator on the right-hand side is stricly positive unless both angular momentum $J_X$ and $J_Y$ vanishes.
\end{corollary}

\begin{proof}
If both $J_X$ and $J_Y$ vanishes, the inequality (with $C=0$) is the simply inequality \eqref{ineq:murw}. If $J_X=0$ but $J_Y \neq 0$, then \eqref{ineq:fin} follows using that 
$$0 \le e^{2U} \le e^{2 |U|} \le e^{ C R^{1/2}},$$
using Lemma \ref{lem:reu}.
If $J_X \neq 0$, we may first drop the positive term containing $J_Y$ to obtain
\begin{eqnarray*}
\mu_R
&\le& \frac{ 2 \xid^R e^{2 \rho}}{K^2+ J_X^2 e^{-2U}+(J_Y-A J_X)^2 e^{2 U }R^{-2}}  \, \frac{d\mu}{ds} \\
&\le& \frac{ 2 \xid^R e^{2 \rho}}{K^2+ J_X^2 e^{-2U}}  \, \frac{d\mu}{ds}
\end{eqnarray*}
and we proceed as in the previous case, using that $J_X \neq 0$. 
\end{proof}


\subsection{Completion of the proof}

We can now establish the geodesic completeness property. 

\begin{lemma} One has 
$R(\xi(s)) \to +\infty$ when $s \to s_1$.
\end{lemma}

\begin{proof}  
Assume that $R(\xi(s)) < +\infty$. Then, $\xi(s)$ stays in a compact region of $\Mcal$ in which, by continuity, $U$, $A$, $\eta$ are uniformly bounded. 
Moreover, from the above analysis,
$$
\left| \frac{d}{ds} \left( R^{-1/2} \xid^R \right)\right| \le |R^{1/2} \Gamma^R_{ab} \xid^a \xid^b| \le F(R(\xi(s)) \left( \xid^R + \frac{d \mu}{ds} \right)
$$
where $F(R)$ is a positive continuous function of $R$, and therefore, is uniformly bounded on any compact interval. By Gronwall lemma, it follows that $\xid^R$ is uniformly bounded and since the curve is timelike, we obtain uniform bounds on all components of $\xid$ and then $W^{2,1}$ bounds on $\xi$ from the previous inequality. In particular, the curve is uniformly timelike for the flat quotient metric as in Lemma \ref{lem:qut}, i.e. there exists a constant $D > 0$ such that for all $s \in [s_0,s_1]$, 
$$
\xid^R \ge D + |\xid^\theta|.
$$
Now either $s_1 <+\infty$, in which case, in view of the previous estimates, one can continue the solution beyond $s_1$, contradicting the fact that the curve is maximal, or $s_1=+\infty$, but then the previous inequality implies, after integration in $s$, that $R(\xi(s)\rightarrow +\infty$ as $s\rightarrow+\infty$.
\end{proof}

\noindent\emph{Completion of the proof of Theorem \ref{theo:1}.} Since 
$R(\xi(s))-R(\xi(s_0))=\int_{s_0}^{s} \xid^R(s) ds$, it follows from the previous lemma that a uniform upper bound such as $\xid^R < C $ implies that $s_1=\infty$. The same is true if we do not have a uniform upper bound on $\xid^R$, but the bound 
$\xid^R(s) \le C R^{p}(s)$  
for some $C> 0$ and $p \in [0,1)$. 

Our estimate on $\mu_R=\rho_R+\frac{1}{4R}$ yields 
$$
\aligned
\left| R^{-1/2}\Gamma^R_{ab} \xid^a \xid^b \right|
&\le R^{-1/2} C e^{-2\rho} \left( (\mu_R)^{1/2} +1 \right) C e^{C R^{1/2}} 
\\ 
&\le C R^{-1/2} e^{-2\rho+C R^{1/2} } + C R^{-1/2} e^{-\rho} (\xid^R)^{1/2}\frac{d}{ds}\left( \rho +\frac{1}{4R} \right)^{1/2}.
\endaligned
$$
The first term on the right is uniformly bounded in view of Corollary~\ref{lem:ring17w}, and therefore integrable provided that $s_+ < +\infty$. For the second term, we apply the Cauchy-Schwarz inequality to $(\xid^R)^{1/2}\frac{d}{ds}\left( \rho +\frac{1}{4R} \right)^{1/2}$ and  estimate
$$
\int_{s_0}^{s}R^{-1/2} e^{-\rho} \xid^R  ds\le \int_{s_0}^{s}C e^{-C R} \xid^R ds \le  C e^{-C R(s_0)}, 
$$
using again Corollary~\ref{lem:ring17w}. Finally, we also have
$$
\aligned
&\int_{s_0}^s e^{-\rho} \, R^{-1/2} \frac{d}{ds}\left( \rho +\frac{1}{4} \ln R \right) \, ds 
\\
&\leq 
R^{-1/4}(s_0) 
\int_{s_0}^s e^{-\rho} \, R^{-1/4} \frac{d}{ds}\left( \rho +\frac{1}{4} \ln R \right) \, ds 
 \leq 
e^{-\rho}R^{-1/2}(s_0), 
\endaligned
$$
which is uniformly bounded. This implies that $|\xid^R| \leq C \, R^{1/2}$ and completes the proof of Theorem \ref{theo:1}.


\section*{Acknowledgments}

This paper was completed when the second author (PLF) was spending the Fall Semester 2013 at the Mathematical Sciences Research Institute (MSRI), Berkeley, thanks to a financial support from the National Science Foundation under Grant No. 0932078 000.  The authors were supported by the Agence Nationale de la Recherche through the grant ANR SIMI-1-003-01 (Mathematical General Relativity. Analysis and geometry of spacetimes with low regularity). 


\end{document}